\documentclass[11pt,leqeq]{article}
\usepackage{amssymb,amsthm}
\usepackage{amsmath}
\usepackage{amscd}
\usepackage{epsf,graphics,graphicx,psfrag}
\usepackage{graphics,graphicx,psfrag}
\usepackage{amscd}
\usepackage{epsfig}
\usepackage[all]{xy}
\newtheorem{thm}{Theorem}
\newtheorem{lema}[thm]{Lemma}
\newtheorem{cor}[thm]{Corollary}
\newtheorem{prop}[thm]{Proposition} 

\newtheorem{defi}[thm]{Definition}

\date{}

\begin{document}
\setlength{\baselineskip}{16pt}
\title{Invariants from classical field theory}
\author{Rafael D\'\i az and Lorenzo Leal}
\maketitle

\begin{abstract}
We introduce a method that generates invariant functions from
perturbative classical field theories depending on external
parameters. Applying our methods to several field theories such as
abelian $BF$, Chern-Simons and $2$-dimensional Yang-Mills theory,
we obtain, respectively, the linking number for embedded
submanifolds in compact varieties, the Gauss' and the second
Milnor's invariant for links in $S^3$, and invariants under
area-preserving diffeomorphisms for configurations of immersed
planar curves.
\end{abstract}

\section{Introduction}

Suppose we have a physical system determined by an action
functional $S$ depending on fields as usual, but also on external
parameters $p$ belonging to some space $P$.  Fixing the external
parameters we can study our field theory from a classical or
quantum point of view.  A fundamental problem in physics is to
understand how the system changes as $p$ varies. In this work we
focus our attention on a particular issue arising within this
setting. Suppose that fields and external parameters are acted
upon by a Lie group $G$, and that $S$ is invariant under the
simultaneous action of $G$ on fields and parameters. How is the
group of symmetries $G$ reflected on the states of the system as
$p$ varies? The answer to this problem is very different depending
on whether we look at it from a quantum or a classical point of
view. The quantum situation requires techniques related to
anomalies
and will not be discussed in this work.\\

We look at the classical situation from two points of view: we
consider exact as well as perturbative solutions of the equations
of motion. The exact results admit a simple and clean description
given in Section \ref{ER}. Simplifying slightly, the conclusion is
that associated with such a classical system there is a function
$S_{os}:P \longrightarrow \mathbb{R}$ whose value at $p$ is
obtained by evaluating $S$ on-shell, i.e.
$S_{os}(p)=S(\varphi(p),p)$ where $\varphi(p)$ represents a
solution of the equations of motion of the action $S(
\ \ , p)$. If the correspondence $p \longrightarrow
\varphi(p)$ can be constructed in a $G$-equivariant fashion, then
the function $S_{os}:P
\longrightarrow \mathbb{R}$ is invariant under the action of
$G$. We apply this technique  to abelian $BF$ theory on compact
oriented manifolds, with the external parameters being pairs of
non-intersecting embedded submanifolds of the appropriated
dimensions. The on-shell action is the linking number of the
embedded submanifolds, an invariant under diffeomorphism of the
ambient manifold connected to the identity. We write explicitly
using local coordinates the on-shell
action for this example.\\

It is seldom possible to find explicitly the solutions of the
equations of motion for most interesting Lagrangian systems. In
this case an alternative route is to look for perturbative
solutions of the equation of motion. Section
\ref{poa} contains the main result of this work: the proof that a
hierarchy of invariant functions on parameter space can be
obtained from the computation of the perturbative on-shell action
$S_{os}$ of a functional action $S$ invariant under the
simultaneous action of a Lie group on fields and external
parameters. In order to show this result we develop in Section
\ref{PS} a fairly explicit model for the study of perturbative
solutions of systems of equations (algebraic, differential,
integral, etc.) We show that, in the non-degenerate case, the
perturbative solution of the system is unique. In the degenerated
case, even though uniqueness is lost, our methods guarantee that
the space of perturbative solutions is non-empty, and provide an
explicit solution under weak assumptions. We apply this technique
to solve a couple of general equations arising in the context of
Hodge algebras using techniques closely related to homological
perturbation theory \cite{gu}. In Section
\ref{csli} we consider an example of interest in low dimensional
topology, namely, we applied our methodology to Chern-Simons-Wong
action and show that it yields link invariants, in particular, we
obtain the Gauss' and the second Milnor's invariants. In Section
\ref{ym} we discuss how invariants under area preserving
diffeomorsphisms of $\mathbb{R}^2$ can be obtained given a generic
finite family of immersed curves in the plane, applying our
methods to Yang-Mills-Wong action. The Wong term that we add both
in the Chern-Simons and Yang-Mills theories may be thought,
physically, as the action for conservation of cromo-electric
charge, and mathematically, as the action functional for parallel
transport in a fiber bundle. In Section
\ref{dp} we present a brief discussion of open
problems and future lines of research.\\

\section{Exact results}\label{ER}

Fix spaces $F$ and $P$, thought as the space of fields and
parameters, respectively. Assume that a Lie group $K$, thought as
the gauge group, acts on $F$. Fix another Lie group $G$ which acts
on $F$ and $P$, together with a map $k: G
\longrightarrow K$. Let $G$ act on $F \times P$ via the diagonal action.
Suppose we have  map $S:F \times P \longrightarrow
\mathbb{R},$ thought as the action of a classical field theory, satisfying
$S(g\varphi,gp)=S(\varphi,p)$ and $S(k\varphi,p)=S(\varphi,p)$ for
$(\varphi,p)
\in F
\times P$, $g \in G$ and $k \in K$. In addition, assume we have
map $\alpha: P \longrightarrow F$ such that $\alpha(gp)=k(g)
g\alpha(p) $ for $g \in G$ and $p \in P$. Lemma \ref{exact1} below
explains how one can get a $G$-invariant function on $P$ from this
data.

\begin{lema}\label{exact1}{\em
The map $S_{\alpha}:P \longrightarrow \mathbb{R}$ given by
$S_{\alpha}(p)=S(\alpha(p),p)$ is $G$-invariant.}
\end{lema}
Indeed for  $g \in G$ and $p \in P$ we have that:
$$S_{\alpha}(gp)=S(\alpha(gp),gp)=S(k(g)g\alpha(p),gp)=S(g\alpha(p),p)=S(\alpha(p),p)=S_{\alpha}(p).$$

A fundamental question thus arises in this context: how can one
obtain such a map $\alpha$? We are going to show via examples that
it is often possible to find a map $\alpha$ with the required
properties by solving the equation of motion, i.e. finding for
each $p \in P$ a solution $\varphi(p)$ of the equation
$$\frac{\partial S}{\partial\varphi}(\varphi(p),p)=0.$$ Thus we are going to show that the so-called on-shell action
$S_{os}(p)=S(\varphi(p),p)$ is a $G$-invariant function on parameter space.\\

Let us first consider abelian $BF$ gauge theory generalizing a
construction of \cite{l2}. Let $M$ be a compact oriented manifold
of dimension $n$ and fix $1 \leq p \leq n$. The space of fields is
$$\Omega^{p}(M)
\oplus \Omega^{n-p-1}(M)$$ where $\Omega^{i}(M)$ denotes the space
of differential $i$-forms on $M$. Let $BE(M,i)$ be the space of
bounding embedded $i$-dimensional submanifolds of $M$, i.e.
$$BE(M,i)= \{\gamma \  | \  \gamma: \Sigma \rightarrow M  \mbox{
bounding embedding, $\Sigma$ a compact oriented $i$-manifold }\}/
\sim.$$ An embedding $\gamma: \Sigma \longrightarrow M$ is bounding if there exists embedding
$\delta:\Delta \longrightarrow M$, where $\Delta$ is an oriented
manifold with boundaries such that $\partial(\Delta)=\Sigma$ and
$\delta|_{\Sigma}=\gamma$. Embeddings $\gamma_1: \Sigma_1
\longrightarrow M$ and $\gamma_2:\Sigma_2
\longrightarrow M$ are $\sim$ equivalent if there exists an orientation preserving diffeomorphism
$\phi: \Sigma_1 \longrightarrow \Sigma_2$ such that $\gamma_1=\phi
\circ \gamma_2$. The space of parameters is $BE(M,p)
\times BE(M,n-p-1)$ and the action functional
$$S:(\Omega^{p}(M) \oplus \Omega^{n-p-1}(M))
\times BE(M,p) \times BE(M,n-p-1) \longrightarrow \mathbb{R}$$
is a $BF$ theory couple to external parameters given by
\begin{equation*} \label{BF}
-S(A_{1},A_{2},\gamma _{1},\gamma _{2}) = \int_{M} A_{1} \wedge
dA_{2} + \int_{\Sigma _{1}} \gamma_{1}^{*}A_{1} + \int_{\Sigma
_{2}}\gamma_{2}^{*} A_{2},
\end{equation*}
for  $(A_{1},A_{2},\gamma _{1},\gamma_{2}) \in (\Omega^{p}(M)
\oplus \Omega^{n-p-1}(M)) \times BE(M,p) \times BE(M,n-p-1).$
The action $S$ is invariant under gauge transformations $A_{1}
\rightarrow A_{1}+df_{1},$ $A_{2} \rightarrow A_{2}+df_{2},$ where
$f_{1} \in
\Omega^{p-1}(M)$ and $f_{2} \in
\Omega^{n-p-2}(M)$. Let $A(M)$ be the infinite dimensional Lie group of
automorphisms of $M$ connected to the identity; $A(M)$ acts on
forms and embedded submanifolds by pull back and push forward,
respectively. $S$ is manifestly
$A(M)$-invariant since it is metric-independent.\\

Recall \cite{gh} that the Poincar\'e dual form $P(\gamma)
\in \Omega^{n-p}(M)$ of an embedding
$\gamma:\Sigma \longrightarrow M $  is uniquely determined, modulo
de addition of an exact form, by demanding that $$\int_{\Sigma}
\gamma^{*}A =\int _{M}P(\gamma)\wedge A$$ for $A
\in \Omega^{p}(M).$ Poincar\'e dual forms have the following
properties: if $\gamma:
\Sigma \longrightarrow M$ is the boundary of  $\delta:\Delta
\longrightarrow M$, then  $d(P(\Delta))=P(\Sigma)$; if
$\phi:M \longrightarrow M$ is a diffeomorphism, then $P(\phi^{-1}
\circ \gamma)=\phi^{*}P(\gamma)$.  Poincar\'e dual forms, among many other
things, are useful to compute the linking number
$lk(\gamma_{1},\gamma_{1})$ of embedded bounding submanifolds
$\gamma_{1}:
\Sigma_{1} \longrightarrow M$ and $\gamma_{2}: \Sigma_{2}
\longrightarrow M$ of dimension $p$ and $n-p-1$, respectively, as
follows: $$lk(\gamma_{1},\gamma_{2})= \int _{M}P(\Sigma_{1})\wedge
P(\Delta_{2})$$ where  $\delta_{2}:\Delta_{2} \longrightarrow M$
is such that $\partial(\Delta_{2})=\Sigma_{2}$ and
$\partial(\delta_{2})=\gamma_{2}.$   Using Poincar\'e dual forms
the action $S$ may be written as:

$$-S(A_{1},A_{2},\gamma _{1},\gamma _{2}) = \int_{M} A_{1} \wedge
dA_{2} + \int_{M} P(\Sigma _{1}) \wedge A_{1} + \int_{M} P(\Sigma
_{2})
\wedge A_{2}.$$
Varying $S$ with respect to $A_{1}$ and $A_{2}$ we obtain the
equations of motion $$dA_{1} = (-1)^{p}P(\Sigma _{2}) \mbox{ \
\ and \ \ } dA_{2} = (-1)^{p(n-p)+1}P(\Sigma _{1}).$$  Thus the on-shell
action $-S_{os}(A_{1},A_{2},\gamma _{1},\gamma _{2})$ is given by
$$
\int_{M} P(\Sigma _{2})
\wedge A_{2}=\int_{\Sigma _{2}} \gamma_{2}^{*}
A_{2}=\int_{\partial(\Delta _{2})} \gamma_{2}^{*}
A_{2}=\int_{\Delta _{2}} \gamma_{2}^{*} (dA_{2})=\int_{M} P(\Sigma
_{1}) \wedge P(\Delta _{2}).$$  We have shown that the on-shell
action is given by
$$-S_{os}(\gamma _{1},\gamma _{2})=\int_{M} P(\Sigma _{1})
\wedge P(\Delta _{2}).$$
From this expression it is clear that the on-shell action $S_{os}$
is an $A(M)$-invariant function on $BE(M,p)\times BE(M,n-p-1)$,
indeed $S_{os}(\gamma _{1},\gamma _{2})$ computes the linking
number of the embedded submanifolds $\gamma _{1},\gamma _{2}.$
Let us consider the case where $M=\mathbb{R}^{n}$ and use
coordinates $(x_1,x_2,...,x_n)$ to write the solution of the
equations of motion and the on-shell action. The Poincar\'e dual
form $P(\gamma)=P(\gamma)_{\mu}dx_{\mu}$ of an embedded
$p$-manifold $\gamma:\Sigma \longrightarrow
\mathbb{R}^{n}$ is given by $$P(\gamma)_{\mu}(x)= \varepsilon_{\mu,\nu} \int_{\Sigma}
\gamma^{*}(dx^{\nu}) \delta
  ^{n}(x-\gamma(a)),$$
where $a \in \Sigma$, $\delta^{n}(x-\gamma(a))$ is the
$n$-dimensional  Dirac's delta function centered at $\gamma(a)$,
and for $\mu=(\mu_1,...,\mu_p)$ we set $dx^{\mu}=dx^{\mu_1} \wedge
... \wedge dx^{\mu_n}$. The solution
$A_{\gamma}=A_{\gamma,\mu}dx^{\mu}$ of the equations of motion is
given by $$A_{\gamma, \mu}(x) = \int _{R^{n}} dy^{n}
 \epsilon _{\mu c \nu}
 \,\frac{(x-y)^{\mu_{c}}}{\mid x-y \mid
 ^{n}}P(\gamma_{2})_{\nu}(y),$$
where $|x|= \sqrt{x_{1}^{2}+x_{2}^{2}+...+x_{n}^{2}}$ and
$\epsilon_{\mu c \nu}=\epsilon _{\mu_{1} \mu_{2} ...\mu_p c
\nu_1... \nu_{n-p-1}}$ is the completely antisymmetric symbol in
$n$-dimensions. Using the expression above for both $A_{1}$ and
$A_{2}$, we obtain explicitly the on-shell action $S_{os}$, i.e.,
the linking number of the embeddings $\gamma _{1}$ and $\gamma
_{2}$:

 \begin{equation*}
S_{os}(\gamma_1,\gamma_2) =
\int_{\gamma_{1}}\int_{\gamma_{2}}\epsilon _{\mu c \nu}
\gamma_{1}^{*}(dx^{\mu}) \gamma_{2}^{*}(dx^{\nu})
\frac{(\gamma_1(a)-\gamma_1(b))^{c}}{\mid \gamma_1(a)-\gamma_1(b) \mid
 ^{n}}.
\end{equation*}

For $n=3$, $p=1$,  $\gamma _{1}$ and $\gamma _{2}$ are actually
closed curves in $\mathbb{R}^{3}$ and the on-shell action
$S_{os}(\gamma_1,\gamma_2)$ is  the Gauss' linking number. We may
also be consider the case $n=1$, $p=0$, indeed let us consider a
slightly generalized action. The space of fields
${C^{\infty}(\mathbb{R})}^{ n}$ consists of $n$-tuples
$(f_{1},...,f_{n})$ of piecewise smooth functions on the real line
$\mathbb{R}$. The manifold of parameters
$C_{n}(\mathbb{R})=\{(x_1,...,x_n) \in {\mathbb{R}}^n \
\ | \ \ x_i \neq x_j \mbox{ for } i \neq j \}$ is the space of
configurations of $n$ distinguishable points on the real line.
Configuration space and its compactification are studied in
\cite{axel, bott}. The action $S:{C^{\infty}(\mathbb{R})}^{\times
n}
\times C_{n}(\mathbb{R})
\longrightarrow \mathbb{R}$ given by
\begin{equation*} \label{aconfig2}
S(f_1,..,f_n;x_1,..,x_n)=\sum_{i<j}\int_{\mathbb{R}}f_{i}(x)f^{\prime}_j(x)dx
- \sum_{i}f_{i}(x_i)
\end{equation*}
is invariant under under the natural action of the group
$A(\mathbb{R})$  of orientation preserving diffeomorphisms of
$\mathbb{R}$. The equations of motion are $-
\sum_{j<i}f_{j}^{\prime} + \sum_{i<j}f_{j}^{\prime} -
\delta_{x_i}=0,$ where $\delta_{x}$ is the delta function concentrated in $x$  and $1
\leq i\leq n$. Integrating we get $-
\sum_{j<i}f_{j} + \sum_{i<j}f_{j}=\theta_{x_i},$ where
$\theta_{x}$ denotes the Heaviside theta function with jump at
$x$. Thus the equation of motion is $Af=\theta$, where $A$ is the
$n
\times n$ matrix given $A_{i,j}=sg(j-i),$  $f=(f_1,...,f_n)$ and $\theta=(\theta_{x_1},...,\theta_{x_n})$.
 One can check that $f=B\theta$ where $B_{i,j}=(-1)^{|i-j|}$ and that the
on-shell action $S_{os}:C_{n}(\mathbb{R}) \longrightarrow
\mathbb{R}$ is given by
$$S_{os}(x_1,...,x_n) = \sum_{i<j,k,s}(-1)^{|i-k| + |j-s| }\theta_{x_s}(x_k).$$

Notice that $S_{os}(x_1,...,x_n)$ is indeed an
$A(\mathbb{R})$-invariant function on configuration space as it
should according to Lemma \ref{exact1}.

\section{Perturbative solutions}\label{PS}

Suppose one is interested in finding solutions of an equation of
the form $O(\varphi)= \psi$, where $V$ is a vector space, $O:V
\longrightarrow V$ is a non-necessarily linear map, $\psi$ is an element of $V$, and
$\varphi$ is the unknown. We are actually going to work
perturbatively, so we may as well start with a map $O:V
\longrightarrow V[[\lambda]]$, thus
$O=\sum_{n=0}^{\infty}O_{n}\lambda^{n}$ where $O_{n}:V
\longrightarrow V$ is a non-necessarily linear map. Assume that each $O_n$ admits a globally converging Taylor
expansion $$O_{n}=
\sum_{k=1}^{\infty}O_{n,k}(\varphi,...,\varphi),$$
where $O_{n,k}:V^{\otimes k} \longrightarrow V$ is a multilinear
operator and $O_{n,1}=0$ for $n\geq1$.  Finding solutions of the
equation $\sum_{n=0}^{\infty}O_{n}\lambda^{n}=
\psi$ is a notoriously difficult problem, and no general answer
should be expected. Remarkably, it can be treated perturbatively
as follows: performing the substitutions $\varphi
\rightarrow \lambda \varphi$ and $\psi
\rightarrow \lambda \psi,$ the equation $\sum_{n=0}^{\infty}O_{n}\lambda^{n}=
\psi$ becomes

\begin{equation*}
\sum_{n\geq0,k\geq1}O_{n,k}(\varphi,...,\varphi)\lambda^{n+k-1}=\psi.
\label{pert1}
\end{equation*}
Making the expansion $\varphi=\sum_{i=0}\varphi_{i}{\lambda}^{i}$
transforms the equation above, a system of non-linear equations,
into a system of infinitely many linear equations. Indeed, taking
into account the powers of $\lambda$ we get a zero order equation:
\begin{equation}\label{pert4}
O_{0,1}(\varphi_{0})=\psi,
\end{equation} and for  $n \geq 1$ we get higher order
equations:

\begin{equation*}
\sum_{m,k,i_{1},...,i_{k}}O_{m,k}(\varphi_{i_{1}},...,\varphi_{i_{k}})=0,  \label{pert5}
\end{equation*}
where the sum runs over  non-negative integers
$m,k,i_{1},...,i_{k}$ such that $n=m+
\sum_{s=1}^{k}i_{s} + k -1.$
The equation above may be written in the suggestive form

\begin{equation}
O_{0,1}(\varphi_{n})= - \sum_{m,k
,i_{1},...,i_{k}}O_{m,k}(\varphi_{i_{1}},...,\varphi_{i_{k}}),\label{pert6}
\end{equation}
where  $m \geq 0$, $k \geq 2$ and $n=m+
\sum_{s=1}^{k}i_{s} + k -1.$ Thus necessarily integers $i_{1},
...,i_{k}$ are strictly less than $n$, and if $O_{0,1}$ is
invertible then $(\ref{pert6})$ uniquely determines $\varphi_{n}$
in terms of  $\varphi_{i}$ with $i < n$.\\

In order to find $\varphi_n$ explicitly we need several
combinatorial notions \cite{Bergeron}. A directed graph is a
triple $(V,E,(s,t))$ where $V$ and $E$ are finite sets -- the set
of vertices and edges -- and $(s,t):E\longrightarrow V
\times V$ is a map.  A path $\gamma$ in a graph is a
sequence of edges $e_1,e_2,...,e_k$ such that $t(e_i)=s(e_{i+1})$
for $1\leq i \leq k-1.$ We say that $\gamma$ is path from $s(e_1)$
to $t(e_k).$  A rooted tree $T$ is a directed graph with a
distinguished vertex $r$, called the root, such that for each
vertex $v$ of $T$ there is a unique path in $T$ from $v$ to $r$.
The valence of a vertex $v$ is $val(v)=|star(v)|,$ where
$star(v)=\{e \in E \ \ |\
\ t(e)=v\}.$ Vertex $v$ is a leave if $val(v)=0.$ A vertex that is
not a leave is called internal. The set of internal vertices is
denoted by $V_{I}$ and the set of leaves is denoted by $V_{L}.$
The root of a tree is an internal vertex, except in the case of
the tree $\bullet$ whose unique vertex is the root. A tree
together with a linear order on $star(v)$, for each internal
vertex $v$ is called planar. A labelled planar rooted tree is a
pair $(T,l)$ where $T$ is a planar rooted tree and
$l:V_{I}(T)\longrightarrow
\mathbb{N}$ is a map, called the labelling of $T$. Consider the category whose
objects are labelled planar rooted trees. A morphisms is a pair
$(f, g)$ where $f: V_{T_1} \longrightarrow V_{T_2}$ and
$g:E_{T_1}\longrightarrow E_{T_2}$ are maps such that
$(s_2,t_2)\circ g= (f,f) \circ (s_1,t_1);$ moreover we demand the
pair $(f,g)$ preserves both the label and the linear ordering on
$star(v)$ for each internal vertex $v$.

\begin{defi}\label{rtree}{\em
For $n \geq 1$ let $\textsc{T}_n$ be the set of isomorphism
classes of labelled planar rooted trees $T$ such that: $val(v)\geq
2$ for $v\in V_{I}$ and $\sum_{v \in
V_{I}}(val(v)+l(v))=n+|V_{I}|.$}
\end{defi}

A labelled planar rooted tree $T$ is uniquely constructed by
joining planar subtrees $T_1,...,T_k$, $k
\geq 2$, to the root $r$ labelled by $l$, see Figure $\ref{corollas1}$.
If $T$ is so constructed we  write $T = (T_1,...,T_k)_{l}.$ The
following set theoretical identities hold: $$
V_L(T)=\sqcup_{s=1}^{k}V_l(T_s) \mbox{ \ and \ }
V_I(T)-\{r\}=\sqcup_{s=1}^{k}V_I(T_s).$$
\begin{figure}[h]
\begin{center}\includegraphics[height=2.3cm]{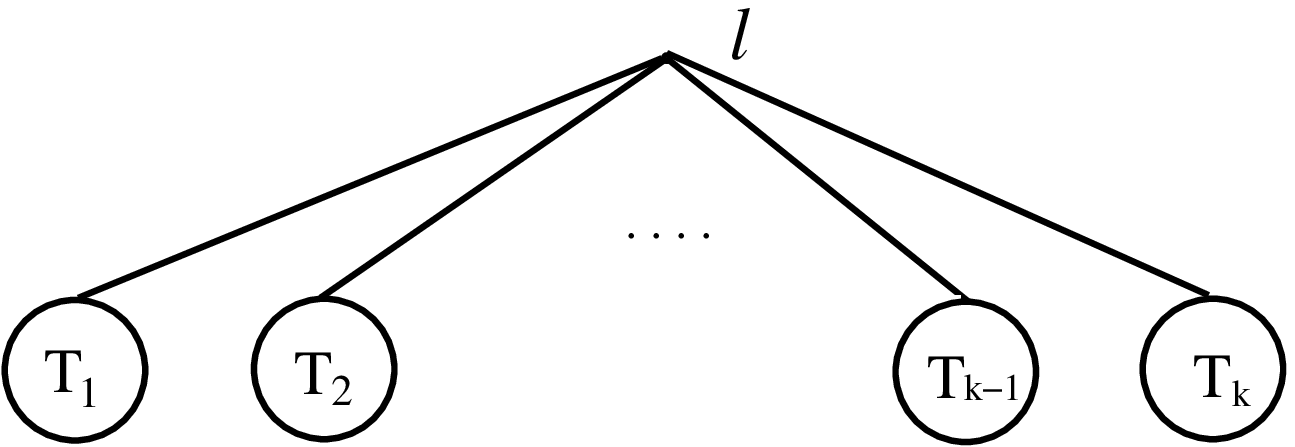}
\caption{\  Tree $(T_1,...,T_k)_l.$ \label{corollas1}}
\end{center}
\end{figure}

\begin{lema}\label{tree1}{\em
Let $T_s$ belong to $\textsc{T}_{i_{s}}$ for $1
\leq s \leq k$. Then $T=(T_1,...,T_k)_l$ belongs to $\textsc{T}_{n}$ for
$n=\sum_{s=1}^{k}i_{s} + k +l -1.$}
\end{lema}

\begin{proof}
The set theoretical identities above  imply that
$$\sum_{v \in V_{I}(T)}(val(v)+l(v))=\sum_{{s},{v \in V_{I}}(T_s)}(val(v)+l(v)) +
val(r_T)+l=$$ $$
\sum_{s=1}^{k}i_{s} + \sum_{s=1}^{k}|V_{I}(T_s)| + k+l
= \sum_{s=1}^{k}i_{s} + k +l - 1 + |V_{I}(T)|.$$

\end{proof}

Next definition assumes that the operator $O_{0,1}$ is invertible.

\begin{defi}\label{tree4}{\em
For $T \in \textsc{T}_n$ let $O_T:V^{\bigotimes |V_{L}(T)|}
\longrightarrow V$ be recursively given by
$$O_{\bullet}=O_{0,1}^{-1} \mbox{ \  and \ } O_{(T_1,...,T_k)_l} =
-O_{0,1}^{-1}(O_{l,k}(O_{T_1},...,O_{T_k})).$$}
\end{defi}

\begin{prop}\label{apert8}{\em
The perturbative solution $\varphi =
\sum_{n=0}^{\infty}\varphi_{n}\lambda^{n}$ of equations $(\ref{pert4})$ and
$(\ref{pert6})$ is given by
\begin{equation}\label{no}
\varphi_n= \sum_{T \in \textsc{T}_{n}}O_{T}(\psi,...,\psi).
\end{equation}}
\end{prop}

\begin{proof}
Let $\varphi_{n}$ be given by $(\ref{no})$ then
$O_{0,1}(\varphi_{n}) = \sum_{T \in
 \textsc{T}_{n}}O_{0,1}(O_{T}(\psi,...,\psi)).$ By  Definition
$\ref{tree4}$, Lemma $\ref{tree1}$, induction and writing  $T =
(T_1,...,T_k)_{l},$ the previous sum  is equal to

$$\sum_{n\geq 0,k \geq 2,T_{1}
\in \textsc{T}_{i_1},...,T_{k}
\in \textsc{T}_{i_k}}-O_{n,k}(O_{T_1}(\psi,...,\psi),...,O_{T_k}(\psi,...,\psi))  =
\sum_{n\geq 0,k \geq
2,i_{1},...,i_{k}}-O_{n,k}(\varphi_{i_{1}},...,\varphi_{i_{k}}).$$
Thus  $\varphi_{n}$ satisfies the required recursion.
\end{proof}

Consider the polynomial equation $a_{n}x^{n} + ... + a_{2}x^{2} +
a_{1}x = y$ with $a_{1} \neq 0.$ Instead of looking for an exact
expression for $x$ as a function of $y$ we look for a perturbative
solution, i.e.  a solution $x =
\sum_{n=0}^{\infty}x_{n}\lambda^{n}$ of the equation
$a_{n}x^{n}\lambda^{n-1} + ... + a_{2}x^{2}\lambda + a_{1}x = y.$
Let $\textsc{T}_{n}^0$ be the set of isomorphism classes of rooted
planar trees with $0$ as the label of all internal vertices.
Proposition $\ref{apert8}$ implies that
$$x_n =  \sum_{T \in
\textsc{T}_{n}^0} (-1)^{|V_i(T)|} a_{1}^{-|V(T)|} a_{T}y^{|V_{l}(T)|},$$
where
$$a_{T}= \prod_{v \in V_i(T)}a_{val(v)}.$$
\begin{cor} \label{pert11}{\em $x = \sum_{n=0}^{\infty}|\textsc{T}_n^0|\lambda^{n}$ is
the formal solution of  $\sum_{n=2}^{\infty} x^{n}\lambda^{n-1} -
x = -1$.}
\end{cor}

Similarly one can check that:

\begin{cor} \label{pert11}{\em $x = \sum_{n=0}^{\infty}|\textsc{T}_n|\lambda^{n}$ is
the formal solution of  $\sum_{n,k\geq 2} x^{n}\lambda^{n+k-1} - x
= -1$.}
\end{cor}

We say that an operator $O:V \longrightarrow V$ has a right
inverse if there exists an operator $P: V \longrightarrow V$ such
that $O(P(\varphi)) =
\varphi$ for $\varphi \in O(V)$. The proof of Proposition
$\ref{apert8}$ yields the following result.

\begin{prop} \label{pert80}{\em Let us assume that  $\psi \in O_{1}(M)$,  $O_{0,1}$ posses a right inverse
$P$, and that  $\sum_{n \geq 0,k
\geq 2,i_{1},...,i_{k}}O_{n,k}(\varphi_{i_{1}},...,\varphi_{i_{k}})
\in O_1(V),$ where $n=m+
\sum_{s=1}^{k}i_{s} + k -1.$
A solution of  (\ref{pert4}) and (\ref{pert6}) is given by
$\varphi_n=
\sum_{T \in \textsc{T}_{n}}O_{T}(\psi,...,\psi)$, where $O_{(T_1,...,T_k)_l} =
-P(O_{l,k}(O_{T_1},...,O_{T_k})),$ $O_{\bullet}=P.$}
\end{prop}

We refer to the conditions of Proposition \ref{pert80}  as the
consistency conditions. To illustrate how Proposition
\ref{pert80} works we solve perturbatively two general equations arising in
Hodge algebras \cite{f, ko, jz}. The methods we use to solve these
equations resemble the techniques of homological perturbation
theory \cite{gu}. Let $(A,d,<
\ ,\ >)$ be a Hodge algebra, i.e.  $(A,d)$ is a differential graded algebra, $<\  ,\  >:A
\otimes A \longrightarrow \mathbb{R}$ is a graded symmetric non-degenerated bilinear
form, and $A$ admits a Hodge decomposition. The adjoint $d^*$ of
$d$ is such that $<da,b> = < a,d^{*}b> $ for $a,b \in A$.  Let
$\Delta = dd^* + d^{*}d$ be the Laplace-Beltrami operator. The
subspace $\mathcal{H}
\subseteq A$ of harmonic elements is  $\mathcal{H} = Ker(\Delta)$.
There is an orthogonal decomposition $A = Im(d) \bigoplus Im(d^*)
\bigoplus \mathcal{H},$ and by Hodge theory  $\mathcal{H}^i$ is canonically
isomorphic to $H^{i}(A)=\frac{Ker(d^{i})}{Im(d^{i-1})}.$ Also
there exists an operator $Q:A \longrightarrow A$ such that $I =
\Delta Q +
\pi_{\mathcal{H}},$ where $I:A \longrightarrow A$ is the identity map
and $\pi_{\mathcal{H}}:A \longrightarrow \mathcal{H}$ is the
orthogonal projection onto $\mathcal{H}$. Moreover setting $G=
d^{*} Q$ we get that $I = G d + d G + \pi_{\mathcal{H}}.$ We look
for a perturbative solution of the equation
\begin{equation}
\Delta(a) + \sum_{n\geq 0,k\geq2}O_{n,k}(a,...,a){\lambda}^{n+k-1} = b,
\label{adg4}
\end{equation}
where $O_{n,k}:A^{\otimes k} \longrightarrow A$ is a linear
operator of degree $1-k$, $O_{n,1}=0$ for $n\geq1$ and $b \in
A^{1}$. If $\mathcal{H}^1=0,$ then $\Delta Q=I$ on $A^1$  and the
consistency conditions of Proposition \ref{pert80} hold.
\begin{prop}\label{pert8}{\em
If $\mathcal{H}^1=0,$ then a perturbative solution
$a=\sum_{n=0}^{\infty}a_n\lambda^n$ of $(\ref{adg4})$ is given by
$a_{n}=
\sum_{T \in \textsc{T}_{n}}O_{T}(b,...,b)$ where
$O_{(T_1,...,T_k)_l} = -Q(O_{l,k}(O_{T_1},...,O_{T_k}))  \mbox{
and } O_{\bullet}=Q.$}
\end{prop}

Next we look for a perturbative solution of an equation of the
form
\begin{equation}
da + \sum_{n\geq 0,k\geq2}O_{n,k}(a,...,a){\lambda}^{n+k-1} = b,
\label{adg1}
\end{equation}
with $a \in A^1$, $b \in A^2,$  $O_{n,k}:A^{\otimes k}
\rightarrow A$ operators of degree $2-k$, $O_{n,1}=0$ for $n\geq1,$
and  $d(b)=0$. Assume that the operators $O_{n,k}$ satisfy the
generalized Leibnitz rule $$d O_{n,k}(a_1,...,a_k)=
\sum_{i=1}^{k}(-1)^{\overline{a_1}+...+\overline{a_{i-1}}}O_{n,k}(a_1,...,d(a_i),...,a_k),$$
for homogeneous elements $a_1,...,a_k \in A$, where $\overline{a}$
denotes de degree of an homogeneous element $a \in A$. Moreover,
assume that the operator $O_{n,k}$ satisfy, for $t \geq 1$, the
quadratic relations  for fixed $n,m,t$:
\begin{equation}\label{ho}
\sum_{k+l=t+1}\sum_{1 \leq i \leq t-l+1}(-1)^{\overline{a_1}+...+\overline{a_{i-1}}}
O_{n,k}(a_1,...,O_{m,l}(a_i,...,a_{i+l-1}),..., a_t)= 0.
\end{equation}

If $\mathcal{H}^{2} \simeq H^{2}(A)=0,$ then $\pi_{\mathcal{H}} =
0$ and $I = G d + d G $ on $A^{2}$ and thus $G$ is a right inverse
of $d$.

\begin{prop}\label{pert8}{\em
If $\mathcal{H}^{2}=0,$ then a perturbative solution of
$(\ref{adg1})$ is given by $a=\sum_{n=0}^{\infty}a_n\lambda^n$
where $a_{n}=
\sum_{T \in \textsc{T}_{n}}O_{T}(b,...,b)$ and
$O_{(T_1,...,T_k)_l} = -G(O_{l,k}(O_{T_1},...,O_{T_k})) \mbox{ and
} O_{\bullet}=G.$}
\end{prop}

\begin{proof}
We have to show that  $\sum_{T \in
\textsc{T}_{n}}d\widetilde{O}_{T}(b,...,b)=0,$ where
$\widetilde{O}_{(T_1,...,T_k)_l} = -O_{l,k}(O_{T_1},...,O_{T_k}).$
Since
$$d(-O_{l,k}(O_{T_1},...,O_{T_k}))=
\sum_{i=1}^{k}\pm O_{n,k}(O_{T_1},...,
dG\widetilde{O}_{T_i},...,O_{T_k})=$$ $$\sum_{i=1}^{k}\pm
O_{n,k}(O_{T_1},..., \widetilde{O}_{T_i},...,O_{T_k})\mp
O_{n,k}(O_{T_1},..., Gd\widetilde{O}_{T_i},...,O_{T_k}).$$ Thus
$\sum_{T \in \textsc{T}_{n}}d\widetilde{O}_{T}(b,...,b)$ equals
the sum of two terms; the first one vanishes by (\ref{ho}) and the
second one by induction.
\end{proof}

Notice the similarity between the conditions of Proposition
\ref{pert8} and the axioms defining $A_{\infty}$-algebras
\cite{k, s, s2}, especially when the operators $O_{n,k}$ vanish for $n >1$.
Indeed our conditions involve a countable family of operators
$O_{n,k}$ satisfying a countable number of quadratic equations. It
would be interesting to investigate the operadic and geometric
interpretation of the conditions of Proposition \ref{pert8}.

\section{Perturbative on-shell action}\label{poa}

A major difficulty in the process of obtaining invariant functions
by evaluating the on-shell action of classical field theories is
that one can seldom find explicitly the solutions of the equations
of motion. We show in this section that one can get around this
problem if we evaluate instead the perturbative on-shell action.
An interesting feature of the perturbative approach is that one
gets automatically a hierarchy of invariants indexed by the
natural numbers. The zero level is obtained by linearization of
the equations of motion. Higher order invariants are obtained
applying a sophisticated recursive procedure, where each step
consists in solving the linear equations of motion with
varying non-homogeneous term.\\

We are ready to discuss the main result of this work. We are going
to show that under suitable conditions, made precise below, if we
are given an action $$S:F \times P \longrightarrow
\mathbb{R}[[\lambda]]$$ then there are infinitely many $G$-invariant functions
$S_{(n)}:P\longrightarrow \mathbb{R}$ with $n \geq 0$ that are
constructed by evaluating the on-shell action $S_{os}$
perturbatively. Let us then proceed to state the conditions
necessary for this result. First, we assume that we have a group
$G$ which acts via a diagonal action on $F \times P.$ Second, we
assume that the space $F$ of fields is provided with a
non-degenerated $G$-invariant symmetric bilinear form $< \  , \
>:F\otimes F
\longrightarrow \mathbb{R}.$ Moreover, we assume that linear operators on $F$ can be written in the form
$<\chi, \  >$ for some $\chi \in F$. Expand $S$ in powers of
$\lambda$ as
$$S(\varphi,p)=\sum_{n=0}^{\infty}S_{n}(\varphi,p)\lambda^{n}$$ and
consider the further expansions
$$S_0(\varphi,p) = \sum_{k=1}^{\infty}\frac{Q_{0,k}(\varphi,...,\varphi,p)}{k}
\mbox{ \ \ and \ for \  $n \geq 1$ \ set \ \ }
S_n(\varphi,p) =
\sum_{k=3}^{\infty}\frac{Q_{n,k}(\varphi,...,\varphi,p)}{k},$$
where $Q_{n,k}:F^{\otimes k}\times P \longrightarrow \mathbb{R}$.
Notice that for  $n \geq 1$ the maps $Q_{n,k}$ are defined for $k
\geq 3$, this assumption fits nicely with the
results of the previous section. Our third assumption is that
$Q_{n,k}(g\varphi,...,g
\varphi,g p)=Q_{n,k}(\varphi,..., \varphi, p)$ for
each $g \in G.$ By the previous assumptions we can write
$Q_{n,k}(\varphi,...,\varphi,\psi,p) =
<O_{n,k-1}(\varphi,...,\varphi,p),\psi>$ where the maps
$O_{n,k}:F^{\otimes k}\times P \longrightarrow F$ are such that
$$O_{n,k}(g
\varphi,...,g
\varphi,g p)=g O_{n,k}( \varphi,..., \varphi, p)$$ for $ g
\in G.$ Set also $Q_{0,1}(\psi,p)=-<j(p),\psi>$ where
$j(gp)=gj(p)$ for $g
\in G$.
 The Euler-Lagrange equations are determined by the identity
\begin{equation*}\label{lag1}
\frac{d}{d\epsilon}S(\varphi + \epsilon \psi)|_{\epsilon=o}=
Q_{0,1}(\psi) + Q_{0,2}(\varphi,\psi)+
\sum_{n\geq0 ,k\geq
3}Q_{n,k}(\varphi,...,\varphi,\psi)\lambda^{n}.
\end{equation*}
By the previous assumptions the critical points of $S$ are the
solutions of the equation
\begin{equation*}\label{lag1}
O_{0,1}(\varphi,p) + \sum_{n \geq 0,
k\geq2}O_{n,k}(\varphi,...,\varphi,p)\lambda^{n}= j(p).
\end{equation*}
Making $\varphi \longrightarrow \lambda \varphi,$ $j
\longrightarrow \lambda j$, the critical points of $S$ are
determined by the equation
\begin{equation}\label{elag1}
O_{0,1}(\varphi,p) +
\sum_{n\geq 0,k\geq2}O_{n,k}(\varphi,...,\varphi,p)\lambda^{n+k-1}=
j(p).
\end{equation}

\begin{prop} {\em If $O_1(\ \ , p)$ is invertible for each $p \in P$, then
the perturbative solution $\varphi(p) = \sum_{n=0}^{\infty}
\varphi_n(p) \lambda^n$ of $(\ref{elag1})$ is such that
$\varphi(gp)= g \varphi(p)$ for $g \in G$.}
\end{prop}

\begin{proof}
We show that $\varphi_{n}(gp)= g \varphi_{n}(p)$ for $g
\in G.$  From Proposition $\ref{apert8}$  and Lemma $\ref{epert3}$ we get
$$\varphi_{n}(gp)= \sum_{T \in\textsc{T}_{n}}O_{T}(j(gp),...,j(gp),gp)= g  \sum_{T \in
\textsc{T}_{n}}O_{T}(j(p),...,j(p),p) = g \varphi_{n}(p).$$
\end{proof}

\begin{lema}\label{epert3}{\em
$O_{T}(g \alpha,...,g \beta, gp)=g O_{T}(\alpha,...,\beta,p)$ for
$g
\in G$ and $\alpha,...,\beta \in F$.}
\end{lema}

\begin{proof}
Assume that $T =(T_1,...,T_k)_l$,  then
$$O_{(T_1,...,T_k)_l}(g \alpha,...,g \beta, g p)=
-O_{0,1}^{-1}(O_{l,k}(O_{T_1}(g \alpha,...,g \kappa),
...,O_{T_k}(g
\tau,...,g \beta), g p),g p),$$ which by induction is equal to $$-g
O_{0,1}^{-1}(O_{l,k}(O_{T_1}( \alpha,..., \kappa),
...,O_{T_k}(\tau,...,\beta),p),p) = g O_{(T_1,...,T_k)_l}(
\alpha,..., \beta, p).$$
\end{proof}

Similarly one can prove the following result.

\begin{prop} \label{pert10}{\em Assume that  $j(p) \in O_{1}(V,p)$, $O_{0,1}(\ \ ,p)$ has a right inverse
$P(\ \ ,p)$ and $\sum_{n,k \geq
2,i_{1},...,i_{k}}O_{n,k}(\varphi_{i_{1}},...,\varphi_{i_{k}},p)
\in O_{0,1}(V,p)$ where $n=m+
\sum_{s=1}^{k}i_{s} + k -1.$
The solution $\varphi_n(p)=
\sum_{T \in \textsc{T}_{n}}O_{T}(j,...,j,p)$ of  $(\ref{elag1})$ satisfies
$\varphi_n(gp)=\varphi_n(p)$ for $g \in G$.}
\end{prop}

According to Propositions  $\ref{pert10}$ if $O_{0,1}$$(\
\ ,p)$ has a right inverse, then
$\varphi(p) =
\sum_{n=0}^{\infty}\varphi_{n}(p) \lambda^n$ given by $(\ref{no})$
is a perturbative solution of $(\ref{elag1})$. Plugging this
solution in $S$ we obtain that the perturbative on-shell action
$S_{os}:P \longrightarrow
\mathbb{R}[[\lambda]]$ which is given by $$S_{os}(p) =
S(\varphi(p),p)=\sum_{n=0}^{\infty}S_{(n)}(p)\lambda^{n}.$$ We
proceed to show that the functions $S_{(n)}:P
\longrightarrow \mathbb{R}$ are $G$-invariant.
For $n\geq0$ let $\textsc{R}_n$ be the set of isomorphisms classes
of labelled planar rooted trees $T$ that can be written as
$T=(T_1,...,T_k)_l,$ where $T_s \in \textsc{T}_{i_s}$,
$\sum_si_{s}+l=n,$ for $1 \leq s \leq k$ and $k\geq1$ if $l=0$,
and $k\geq2$ if $l\geq1.$

\begin{defi}{\em
For $T \in \textsc{R}_n$ let $Q_{T}:F^{\otimes |V_{L}(T)|}
\longrightarrow \mathbb{R}$ be given by
$$Q_{(T_1,...,T_k)_{l}} = Q_{l,k}(O_{T_{1}},...,O_{T_{k}}).$$}
\end{defi}
The proof of the following result is similar to that of
Proposition $\ref{apert8}$.
\begin{prop}\label{listo}{\em
$S_{(n)}(p)=\sum_{T \in \textsc{R}_n}Q_{T}(j(p),...,j(p),p)$ for
$n\geq0.$}
\end{prop}

We are finally ready to state and prove the main result of this
paper.

\begin{thm} \label{elag7}{\em
$S_{(n)}:P \longrightarrow \mathbb{R}$ is a $G$-invariant function
for $n\geq0$.}
\end{thm}
\begin{proof} If $p \in P$ and $g \in G$ then
$$S_{(n)}(gp)=\sum_{T \in \textsc{R}_n}Q_{T}(j(gp),...,j(gp),gp)=
\sum_{T \in \textsc{R}_n}Q_{T}(j(p),...,j(p),p)=S_{(n)}(p).$$
\end{proof}

\section{Chern-Simons-Wong theory and link invariants}
\label{csli}

The relation between Chern-Simons theory and link invariants was
first study in \cite{witten} and is by now a solid theory
\cite{D, kauf, tu, rozansky}, studied from a variety of  points of
view. A common feature of these approaches is that they work at
the quantum level. It was proposed in \cite{l} that it is possible
to construct link invariants from perturbative classical
non-abelian Chern-Simons action with an extra term due to Wong
\cite{balachandran, wong}. Our desire to understand the mathematical foundations
underlying the methodology of
\cite{l} was the primary motivation for
this work. The results of this section illustrate the full power
of Theorem \ref{elag7}, which yields a hierarchy of invariant
functions starting from functional actions depending equivariantly
on external parameters. Let $S^3$ be the unit $3$-sphere and
$\mathfrak{g}$ be the Lie algebra of a compact semi-simple Lie
group $G$.  Fix a symmetric non-degenerated bilinear form $Tr$ on
$\mathfrak{g}$ invariant under the adjoint action. The space of
fields $$(\Omega^{1}(S^3)\otimes \mathfrak{g}) \times
M(S^{1},G)^{n}$$ consists of tuples $(a,g_{1},...,g_{n})$ where $a
\in
\Omega^{1}(S^3)\otimes \mathfrak{g}$ is a $\mathfrak{g}$-valued
$1$-form on $S^3$, and $g_{i}:S^{1} \longrightarrow G$ is a
$G$-valued map on the unit circle. The space of parameters
$E_n(S^{1},S^3)$ consists of $n$-tuples
$(\gamma_{1},...,\gamma_{n})$ such that $\gamma_{i}:S^{1}
\longrightarrow S^3$ is a embedded closed curve in
$S^3$, and the images of the $\gamma_{i}$ are mutually disjoint.
Physically, $a$ represents the gauge potential and $g_{i}$ the
chromo-electric charge of a point-particle undergoing non-abelian
interactions. The trajectories in $S^3$ of these particles are
elements of the parameter space, and we will show that the linking
of these particles is
tested in the process of computing the perturbative on-shell action.\\

Let $A(S^3)$ be the group of automorphisms of $S^3$ connected to
the identity. $A(S^3)$ acts by pullback on $\Omega^{1}(S^3)
\otimes \mathfrak{g}$, trivially on $M(S^{1},G)^{n}$,
and by push-forward on $E_n(S^{1},S^3)$. To construct the action
functional we introduce some notation. The pullback of $a$ to
$S^1$ via $\gamma_{i}:S^{1} \longrightarrow  S^3$  is denoted by
$a_{i}(t),$ where $t$ is the standard coordinate on $S^{1}$. Fix
elements $c_{i}
\in \mathfrak{g}$ and for $g_{i} \in M(S^{1},G)$ for $1 \leq i \leq n$, let the chromo-electric charge
be $c_{i}(t)=g_{i}(t) c_{i}g^{-1}_{i}(t),$ for $t \in S^{1}.$ The
covariant derivative of $g_{i}$ along the $i$-th particle is
$D_{t}g_{i}=
\partial_t g_{i} + \lambda a_{i}(t)g_{i}.$ The action functional is
$$S(a,g_1,...,g_n,\gamma_1,...,\gamma_n)=\int_{\mathbb{R}^{3}}  \, Tr(a\wedge da+
\frac 23 \lambda a^3)+S^{int}(a,g_1,...,g_n,\gamma_1,...,\gamma_n),$$
where
\begin{equation*}
S^{int}(a,g_1,...,g_n,\gamma_1,...,\gamma_n)= \sum_{i=1}^n
\int_{\gamma_{i}}dt \, Tr(k_{i}
g^{-1}_{i}(t)D_{t}g_{i}(t))\label{accion int}
\end{equation*}
corresponds to the interaction of $n$ classical Wong particles
carrying non-abelian charge $\cite{balachandran, wong}$.
Chern-Simons action is invariant under the  group $M(S^3,G)$ of
gauge transformations connected to the identity. The action of $u
\in M(S^3,G)$ on $a \in
\Omega^{1}(S^3) \otimes \mathfrak{g}$ is given by $ a^{u}=u^{-1}au + u^{-1} dt.$ The
action $S^{int}$ is gauge invariant if we set $c_{i}^{u}=c_{i}$
and $g_{i}^{u}=u^{-1}g_{i}$. Non-abelian charges $c_i(t)$
transform in the adjoint representation
$c_{i}(t)^{u}=u(t)^{-1}c_{i} u(t).$ With these conventions
$S$ is an $A(S^3)$-invariant function. \\

The variation of $S$ with respect to $a$ yields the equation
$$F_a=\frac 1 2
\sum_{i=1}^{n}P(\gamma_i, c_i(t)),$$ where  $F_a = da + \frac \lambda 2 [a,a]$ is the curvature of $a$
and $P(\gamma_i, c_i(t))$ is a Poincar\'e dual form defined via
the identity
\begin{equation*}
\int_{S^3}Tr(P(\gamma_i, c_i(t)) \wedge b_{i}) = \int_{\gamma_i}
Tr(c_{i}(t) b_{i}(t)) dt. \label{PoincareNA}
\end{equation*}
The variation of $S$ with respect to $g_{i}$ yields the equation
$D_{t}c_{i}= \dot{c}_{i} + \lambda [a_{i},c_{i}]=0$ of
conservation of non-abelian charges.  Thus $c_{i}({t})= u_{i}({t})
\, c_{i}\,u^{-1}_{i}({t}),$ where $u_{i}(t)=Pexp
\,\,( -\lambda \int_{0}^t a_{i}(s)\, ds
\,\,)$ is the path ordered exponential of the gauge potential $a$
along the curve $\gamma_{i}$. According to our general theory the
on-shell action may be expanded as
$$S_{os}(\gamma_1,...,\gamma_n)=\sum_{m=0}^{\infty}S_{(m)}(\gamma_1,...,\gamma_n)\lambda^n$$
where each $S_{(m)}$  should be an $A(S^3)$-invariant function of
the link $(\gamma_1,...,\gamma_n)
\in E_n(S^{1},S^3).$ Using $g_i(t)=u_i(t)g_i(0)$ one can
check that $S_{os}^{int}(\gamma_1,...,\gamma_n)=0$. Let $\
\ \widehat{}:\mathfrak{g}\longrightarrow End(\mathfrak{g})$ be the adjoint
representation of $\mathfrak{g}$ given by $\widehat{x}(y)=[x,y]$
for $x,y \in \mathfrak{g}.$ From the equation
$D_{t}c_{i}=\dot{c}_{i} + \lambda\widehat{a_{i}}(c_{i}(t))=0$ we
see that $c_i(t)=Pexp(-\lambda\int_{0}^t\widehat{a_{i}}dt)c_i,$
and thus the equation $F_a= \frac 1 2\sum_{i=1}^nP(\gamma_i,
c_i(t))$ becomes
$$da=- \frac \lambda 2   [a,a] + \frac 1 2 \sum_{i=1}^{n}P(\gamma_i)c_i +
\frac 1 2 \sum_{m=1}^{\infty}\sum_{i=1}^{n}P(\gamma_i, c_{m,i}(t))c_i\lambda^m,$$
where $$c_{m,i}(t)= \int_{\Delta_{0,t}^m}\bigwedge_{j=1}^{m}
e_{i,j}^*(\widehat{a}), \ \ \Delta_{0,m}^t=
\{(x_1,x_2,...,x_m)\ \ |  \ \ 0 \leq x_j \leq
t \mbox{ and } x_j \leq x_k \mbox{ if } j \leq k\},$$ and the map
$e_{i,j}:\Delta_{0,t}^m
\longrightarrow S^3$ is given by  $e_{i,j}(x_1,x_2,...,x_m)=
\gamma_i(x_j)$.\\

We look for a perturbative solution $a=\sum_p a_{(p)} \lambda^p$
of the equation of motion. The corresponding recursive system of
linear equations is given by $da_{(0)}= \frac 1 2 \sum_{i=1}^n
P(\gamma_i)c_i,$ and for $p \geq 1$
$$da_{(p)}=-\frac 1 2\sum_{s_1+ s_2=p-1}[a_{(s_1)},a_{(s_2)}] + \frac 1 2 \sum_{m=1}^{p}\sum_{i=1}^{n}P(\gamma_i, c_{m,i}(t))c_i,$$
where $$c_{m,i}(t)=
\sum_{s_1 + \cdots + s_m= p-m}\int_{\Delta_{0,t}^m}\bigwedge_{j=1}^{m}
e_{i,j}^*(\widehat{a}_{(s_j)}).$$ Similarly the perturbative
on-shell action is $S_{os}=\sum_{m=0}^{\infty}S_{(m)}\lambda^n$
where for $m\geq 0$ we have
$$S_{(m)}= \int_{S^3}\sum_{s_1+ s_2=p}Tr(a_{(s_1)}da_{(s_2)})
+ \frac 2 3 \int_{S^3}\sum_{s_1+ s_2+
s_3=p-1}Tr(a_{(s_1)}a_{(s_2)} a_{(s_3)}).$$ Thus $S_{(0)}$ is
given by $S_{(0)}=\int_{S^3}Tr(a_{(0)} da_{(0)}).$ If $\Sigma_i:
D^1 \longrightarrow M$ is such that $\partial(\Sigma_i)=\gamma_i,$
i.e. $\Sigma_i$ a Seifert surface for $\gamma_i,$ then we have
that $a_{(0)}=
\frac 1 2\sum_{i=1}^n P(\Sigma_i)c_i$ and we get
$$S_{(0)}=\frac 1 4
\sum_{i,j=1}Tr(c_ic_j)\int_{S^3}P(\Sigma_i)P(\gamma_j)=\frac 1 4
\sum_{i,j=1}Tr(c_ic_j)lk(\gamma_i,\gamma_j).$$
$S_{(0)}$ is a linear combination of linking numbers, so it is a
link invariant as predicted from our general theory. We proceed to
compute explicitly $S_{(1)}$ which is given by $$S_{(1)}=
\int_{S^3}Tr(2a_{(0)}da_{(1)} + \frac 2 3 a_{(0)}^3). $$ We know that $a_{(0)}=
\frac 1 2\sum_{i=1}^n P(\Sigma_i)c_i$ and
$$da_{(1)}=-\frac 1 2[a_{(0)},a_{(0)}] +
\frac 1 2\sum_{i=1}^{n}P(\gamma_i, c_{1,i}(t))c_i,$$ where  $c_{1,i}(t)=
\frac 1 2\sum_{i=1}^n\int_{\Delta_{0,t}^1}
e_{1,j}^*P(\Sigma_i) \widehat{c}_i.$ Plugging these identities in
the previous expression for $S_{(1)}$ we obtain
$$S_{(1)}=-\frac 1 4\sum_{i,j,k}Tr(c_i[c_j,c_k])\left( \frac 1 3\int_{S^3}P(\Sigma_i)P(\Sigma_j)P(\Sigma_k) +
\int_{\Delta_{0,1}^2} e_{1,j}^*P(\Sigma_k) \wedge  e_{2,j}^*P(\Sigma_i)\ \right).$$
The first summand in the formula above should be clear. The second
summand arises from
$$\frac 1 4\sum_{i,j,k}\int_{S^3}Tr(P(\Sigma_i)c_iP(\gamma_j, \int_{\Delta_{0,t}^1}e_{1,j}^*(P(\Sigma_k)\widehat{c}_k)c_j),$$
or equivalently
$$\frac 1 4\sum_{i,j,k}Tr(c_i[c_k,c_j])\int_{S^3}Tr(P(\Sigma_i))P(\gamma_j,\int_{\Delta_{0,t}^1}e_{1,j}^*(P(\Sigma_k)).$$
By the defining properties of Poincar\'e forms and antisymmetry of
the Lie bracket, the later expression is equal to
$$- \frac 1 4\sum_{i,j,k}Tr(c_i[c_j,c_k])\int_{\Delta_{0,1}^2} e_{1,j}^*P(\Sigma_k) \wedge  e_{2,j}^*P(\Sigma_i)\}.$$

The formula obtained for $S_{(1)}$  is a link invariant with a
crystal clear geometric meaning: the first summand counts triple
intersections of the corresponding Seifert surfaces, the second
summand counts pairs of points $s,t$ in the parametrization of
loop $\gamma_i$, such that $\gamma_i(s) \in
\Sigma_k$ and $\gamma_i(t) \in \Sigma_j.$  In the computation of $S_{(1)}$ we make use of the identity
$$da_{(1)}=-[a_{(0)},a_{(0)}] +
\sum_{i=1}^{n}P(\gamma_i, c_{1,i}(t))c_i,$$ thus we assumed that
the right hand side of this identity is a closed two-form. This
assumption is by no means trivial and does not hold universally.
Indeed it imposes a severe restriction on the type of links for
which the invariant $S_{(1)}$ is well-defined: the linking number
of each pair of loops in the link must vanish. The Borromean rings
is an example of link for which the invariant $S_{(1)}$ is
well-defined and non-vanishing. For a proof of this and other
interesting facts regarding the invariant $S_{(1)}$ the reader may
consult \cite{l}. The reader should notice that $S_{(1)}$ is the
second Milnor's invariant \cite{milnor} for links in $S^3$, and
thus our method provides an interpretation for that invariant
coming from perturbative Lagrangian physics. We expect that the
higher order invariants $S_{(n)}$  correspond with higher order
Milnor's invariants which can be computed using higher order
Massey products \cite{mo}.

\section{Yang-Mills theory and area invariants}\label{ym}

In this section we show that it is possible to obtain invariants
of configurations of immersed curves in the plane from
Yang-Mills-Wong action. To our knowledge results of this type have
seldom been reported -- unlike the relation between links and
Chern-Simons theory -- perhaps because the space of immersed
curves in the plane, considered up to area preserving
diffeomorphisms, has not been deeply studied in the mathematical
literature.  The example consider in this section is studied in
full details in \cite{dfl}, here we only highlight the results of
that paper that are useful to illustrate yet another application
of our
method. \\

The basic settings is quite similar to those for Chern-Simons-Wong
action. The space of fields is $(\Omega^{1}(\mathbb{R}^2)\otimes
\mathfrak{g}) \times M(S^{1},G)^{n}$. The space of parameters
$I_n(S^{1},\mathbb{R}^2)$ consists of $n$-tuples
$(\gamma_{1},...,\gamma_{n})$ such that $\gamma_{i}:S^{1}
\longrightarrow \mathbb{R}^2$ is an immersed closed curve in
$\mathbb{R}^2$, and the images of the $\gamma_{i}$ intersect, if
they do, in transversal double points. The group of symmetries for
the Yang-Mills-Wong action is the group of area preserving
diffeomorphisms of $\mathbb{R}^2.$ As before we fix $c_{i}\in
\mathfrak{g}$ and for $1
\leq i \leq n$ we let $g_{i} \in M(S^{1},G)$. The action functional
is given by
$$S(a,g_1,...,g_n,\gamma_1,...,\gamma_n)= \int_{\mathbb{R}^2}
Tr(F_a\wedge *F_a) +\sum_{i=1}^{n}\int_{\gamma_{i}}d\tau
Tr(k_{i}g_{i}^{-1}(\tau)D_{\tau}g_{i}(\tau)),$$ where $F_a=da +
\frac{\lambda}{2}[a,a]$ and $*$ is the Hodge star
operator. According to our general theory the on-shell action may
be expanded as
$$S_{os}(\gamma_1,...,\gamma_n)=\sum_{m=0}^{\infty}S_{(m)}(\gamma_1,...,\gamma_n)\lambda^n$$
where each $S_{(m)}$ should be a function of
$(\gamma_1,...,\gamma_n)
\in I_n(S^{1},\mathbb{R}^2)$ invariant under area preserving
diffeomorphisms of $\mathbb{R}^2.$ One can compute $S_{(0)}$ and
$S_{(1)}$ and check that they are indeed invariants under area
preserving diffeomorphisms. The invariant $S_{(0)}$ admits the
fairly simple expression
$$S_{(0)}= \sum_{i,j}Tr(c_{i}c_{j})J(\gamma_{i},\gamma_{j}),$$
where the functions $J(\gamma_{i},\gamma_{j})$ have the following
geometric interpretation. A generic immersed curve in
$\mathbb{R}^2$ induces a partition of $\mathbb{R}^2$ into a finite
number of compact blocks and an unbounded block. The function
$J(\gamma_{i},\gamma_{j})$ is the sum of the signed areas of the
intersections of the finite blocks of $\gamma_{i}$ with the finite
blocks of $\gamma_{j}$. In complete analogy with the
Chern-Simons-Wong case, the geometric interpretation of $S_{(1)}$
takes into account no just the areas of the intersection blocks,
but also the order in which the intersection blocks appear for
several, at least three, curves. Again  $S_{(1)}$ is only
well-defined for an appropriated choice of curves. In \cite{dfl}
we describe explicitly three planar curves -- a planar version of
the Borromean rings -- for which $S_{(1)}$ is well-defined and
non-vanishing.

\section{Final remarks}\label{dp}

We introduced a method that yields invariant functions from
classical field theories. In the perturbative regime we actually
obtain a countable hierarchy of invariants. Our construction
leaves many open problems and suggest new lines of research. For
Chern-Simons-Wong action and $2$-dimensional Yang-Mills-Wong
action we are, at this point, only able to compute the first two
invariants of the hierarchy. Though Theorem
\ref{elag7} provides explicit formulae for the higher order invariants,
and our computations suggest that the consistency equations are
satisfied, a rigorous proof is needed. We expect the higher order
link invariants arising  in the computation of the perturbative
Chern-Simons-Wong on-shell action, to be closely related to
Milnor's link invariants \cite{milnor}. \\

We believe our methods can be usefully applied to other classical
field theories. In particular, it may be rewarding to look at
Yang-Mills-Wong action in higher dimensions, it should yield
conformal invariants associated with closed curves in spacetime.
It may also be interesting to apply our methods to the generalized
Chern-Simons action of
\cite{sc}, it should yield invariants related with
Chas-Sullivan product in string topology \cite{catta, cs}. In our
study of perturbative solutions we saw that the invertibility of
the quadratic part of the action plays a fundamental role. Often
the quadratic part is not invertible and new techniques are
required in order to get invariants. One possibility is to
introduce, as in the quantum case, fermionic variables and replace
the action with a new one with invertible quadratic part. Thus, it
is plausible that in the classical perturbative regime, the BRST
and BV procedures may still play a role.  Another possibility
arises when the inverse of the  quadratic part of the action  is
no quite well-defined, but rather a singular operator. In this
case techniques from renormalization
\cite{ck1} may become useful in order to replace
invariants given by ill-defined divergent integrals, by their
renormalized values.   In recent years it has become clear that
many constructions in field theory \cite{BaezDolan, cy, RDEP}, as
well as in other branches of physics and mathematics
\cite{BaezDolan3, Bergeron, Blan, Blan2, kho} admit categorical analogues.
It would be interesting to investigate the
categorical foundations of the method introduced in this paper.\\

\subsection*{Acknowledgment}
Our thanks to Edmundo Castillo, Takashi Kimura and Jim Stasheff.
This work was partially supported by projects G2001000712-FONACIT
and 03.006316.2006-CDCH-UCV. We also thank a couple of anonymous
referees for helpful suggestions and remarks.

\bigskip

\noindent ragadiaz@gmail.com\\
\noindent Grupo de F\'isica-Matem\'atica, Universidad Experimental Polit\'ecnica de las Fuerzas Armadas\\
\noindent Caracas 1010, Venezuela\\

\noindent lleal@fisica.ciens.ucv.ve\\
\noindent Centro de F\'isica Te\'orica y Computacional, Universidad Central de Venezuela\\
\noindent Caracas 1041-A, Venezuela\\

\end{document}